\documentclass[9pt,shortpaper,twoside,web]{ieeecolor}
\usepackage{generic}
\usepackage{cite}
\usepackage{amsmath,amssymb,amsfonts}
\usepackage{algorithmic}
\usepackage{graphicx}
\usepackage{textcomp}
\usepackage{subfig}
\usepackage{epstopdf}
\usepackage{doi}
\def\BibTeX{{\rm B\kern-.05em{\sc i\kern-.025em b}\kern-.08em
		T\kern-.1667em\lower.7ex\hbox{E}\kern-.125emX}}
\newtheorem{assumption}{$\mathbf{Assumption}$}
\newtheorem{lemma}{$\mathbf{Lemma}$}
\newtheorem{theorem}{$\mathbf{Theorem}$}

\newtheorem{definition}{$\mathbf{Definition}$}
\newtheorem{corollary}{$\mathbf{Corollary}$}

\DeclareMathOperator*{\argmin}{arg\,min}

\begin{document}
\title{Enhancing Robustness of Control Barrier Function: A Reciprocal Resistance-based Approach}

\author{Xinming Wang, \IEEEmembership{Member, IEEE}, Zongyi Guo, \IEEEmembership{Member, IEEE}, Jianguo Guo, Jun Yang, \IEEEmembership{Fellow, IEEE}, and Yunda Yan, \IEEEmembership{Senior Member, IEEE}
\thanks{Xinming Wang, Jianguo Guo and Zongyi Guo are with the Institute of Precision Guidance, and Control, Northwestern Polytechnical University, Xi’an 710060, China. (email: wangxm0615@gmail.com, guojianguo@nwpu.edu.cn, guozongyi@nwpu.edu.cn).}
\thanks{Jun Yang is with Department of Aeronautical and Automotive Engineering, Loughborough University, Loughborough LE11 3TU, UK (e-mail: j.yang3@lboro.ac.uk).}
\thanks{Yunda Yan is with Department of Computer Science, University College London, London, WC1E 6BT, UK (e-mail: yunda.yan@ucl.ac.uk).}
\thanks{This work was supported in part by the National Natural Science Foundation of China under Grants 52472419, Grant 92271109, and Grant 52272404.}}	

\maketitle

\begin{abstract}
In this note, a new reciprocal resistance-based control barrier function (RRCBF) is developed to enhance the robustness of control barrier functions for disturbed affine nonlinear systems, without requiring explicit knowledge of disturbance bounds. By integrating a reciprocal resistance-like term into the conventional zeroing barrier function framework, we formally establish the concept of the reciprocal resistance-based barrier function (RRBF), rigorously proving the forward invariance of its associated safe set and its robustness against bounded disturbances. The RRBF inherently generates a buffer zone near the boundary of the safe set, effectively dominating the influence of uncertainties and external disturbances. This foundational concept is extended to formulate RRCBFs, including their high-order variants. To alleviate conservatism in the presence of complex, time-varying disturbances, we further introduce a disturbance observer-based RRCBF (DO-RRCBF), which exploits disturbance estimates to enhance safety guarantees and recover nominal control performance. The effectiveness of the proposed framework is validated through two simulation studies: a second-order linear system illustrating forward invariance in the phase plane, and an adaptive cruise control scenario demonstrating robustness in systems with high relative degree.
\end{abstract}


\section{Introduction}
In recent years, advancements in science and technology have significantly shaped the development of control systems, shifting the focus beyond traditional goals, such as stability and control accuracy, toward explicitly enforcing safety requirements. Safety has become a critical concern in practical applications, particularly in autonomous systems and intelligent robotics, where failure to meet safety standards can lead to catastrophic consequences \cite{ozguner2007systems,zanchettin2015safety}. To address this, various control strategies have been proposed to prevent unsafe or undesirable behavior, drawing from diverse perspectives such as reachability analysis, model predictive control, barrier Lyapunov functions, prescribed performance control, funnel control, and control barrier functions (CBFs) \cite{althoff2021set,rawlings2017model,tee2009barrier,bechlioulis2010prescribed,berger2021funnel,Ames2017,guo2022control}. These approaches provide different strengths in enforcing safety constraints while maintaining performance.

Among them, CBFs have emerged as a powerful and systematic tool for safety enforcement in affine nonlinear systems of the form:
\begin{equation}
	\begin{aligned}
		&\dot{\boldsymbol{x}} = \boldsymbol{f}(\boldsymbol{x}) + \boldsymbol{g}(\boldsymbol{x})\boldsymbol{u},
	\end{aligned}
	\label{sys0}
\end{equation}
where $\boldsymbol{x}\in \mathbb{X} \subseteq \mathbb{R}^n$, $\boldsymbol{u}\in \mathbb{R}^m$, $\mathbb{X}$ is denoted as the admissible compact set for states. The nonlinear functions $\boldsymbol{f}: \mathbb{R}^n\rightarrow\mathbb{R}^n$, $\boldsymbol{g}: \mathbb{R}^n\rightarrow\mathbb{R}^{n\times m}$ in (\ref{sys0}) are known locally Lipschitz functions. In this framework, CBFs impose a control Lyapunov function-like condition \cite{ames2014rapidly} to ensure the forward invariance of the 0-superlevel set of a continuously differentiable function $h:\mathbb{R}^n\rightarrow \mathbb{R}$, defined as $\mathbb{C} := \{ \boldsymbol{x}\in \mathbb{R}^n:h(\boldsymbol{x})\geq0\}$.

Despite their advantages in safety guarantee, conventional CBFs, including both zeroing CBF (ZCBF) and reciprocal CBF (RCBF), are often sensitive to model uncertainties and external disturbances, as modeled by
\begin{equation}
	\begin{aligned}
		&\dot{\boldsymbol{x}} = \boldsymbol{f}(\boldsymbol{x}) + \boldsymbol{g}(\boldsymbol{x})(\boldsymbol{u} + \boldsymbol{d}).
	\end{aligned}
	\label{sys}
\end{equation}
where $ \boldsymbol{d}(t,\boldsymbol{x})$, where $\boldsymbol{d}: \mathbb{R}_+\times\mathbb{R}_n \rightarrow \mathbb{D} \subseteq \mathbb{R}^m$ and $\mathbb{D}$ is its admissible compact set. Ignoring these disturbances and uncertainties can significantly degrade the desired control performance and, in some cases, even violate formal safety guarantees \cite{xu2015robustness}. 

To mitigate this, robust and adaptive extensions of CBFs have been developed. These include worst-case robust formulations \cite{jankovic2018robust}, input-to-state safety based CBFs \cite{alan2021safe}, adaptive CBFs for parametric uncertainties \cite{taylor2020adaptive}, and disturbance observer-based compensation methods \cite{DasarXiv,sun2024safety,chen2023robust,takano2020robust,WangarXiv,wang2024tac}. While effective under certain conditions, these approaches often require prior knowledge of the disturbance characteristics, such as upper bounds, Lipschitz constants, or time-derivative information, to mitigate the impact of disturbances or estimation errors. In practice, such information may be difficult to obtain or estimate accurately, which complicates controller design and parameter tuning and limits scalability in practical scenarios. These limits motivate the following fundamental question:

\textit{Is it possible to design a control barrier function that can autonomously adapt to disturbances without requiring specific prior knowledge or modelling of them?}

To address the challenge of enforcing safety under unknown disturbances, we propose a reciprocal resistance-based control barrier function (RRCBF) that enhances robustness without requiring explicit disturbance information. By integrating a reciprocal resistance term, derived from the safe function $h(\boldsymbol{x})$, into the classical ZCBF framework, the proposed method effectively dominates the influence of disturbances near the boundary of the safety set while preserving the conventional forward invariance property in the interior. Specifically, we begin by introducing the foundational reciprocal resistance-based barrier function (RRBF), establishing its forward invariance under bounded disturbances, and comparing it with conventional ZBF and RBF formulations. The framework is then extended to formulate the RRCBF and its higher-order variants. To further enhance performance under complex disturbances, the approach is improved using disturbance estimation and compensation techniques. To show the effectiveness, two simulation examples are considered: a disturbed second-order linear system, used to illustrate and compare the forward invariant sets of ZBF, RBF, and the proposed method; and an adaptive cruise control problem, which highlights the safety performance of the proposed approach in a system with high control relative degree.


Notations: The sets of real numbers, positive real numbers and nonnegative integers are denoted as $\mathbb{R}$, $\mathbb{R}_+$ and $\mathbb{N}$. For $i,j\in\mathbb{N}$ satisfying $j\leq k$, define $\mathbb{N}_{j:k}\triangleq\{ j,j+1,\cdots,k \}$ as a subset of $\mathbb{N}$. Denote $\Vert \boldsymbol{x} \Vert$ as the 2-norm of vector $\boldsymbol{x}$. A continuous function $\alpha:[0, \infty) \rightarrow [0,\infty)$ is said to belong to class $\mathcal{K}$ function, if it is strictly increasing and $\alpha(0)=0$. A continuous function $\alpha_e:(-\infty, \infty) \rightarrow (-\infty,\infty)$ is said to belong to extended class $\mathcal{K}$ function, if it is strictly increasing and $\alpha_e(0)=0$.

\section{Reciprocal resistance-based barrier function}
In this section, the reciprocal resistance-based barrier function is systematically introduced, covering its motivation, definition, robustness to bounded disturbances and comparison with ZBF and RBF.
\subsection{Motivation}
To achieve robustness of external disturbances without requiring additional prior information, our design is inspired by utilizing $h(\boldsymbol{x})$ to construct a reciprocal-like function to increase the resistance in the presence of disturbances/uncertainties. Motivated by this idea, consider the following nonlinear system
\begin{equation}\label{ode}
	\dot{z} = -\alpha z + \frac{\beta}{z},
\end{equation}
where $z\in\mathbb{R}_+$ is the state, $\alpha, \beta\in\mathbb{R}_+$ are parameters to be selected. It is clear that (\ref{ode}) is locally Lipschitz continuous with $z$ for $z\in\mathbb{R}_+$. The following analysis investigates the system’s stability and positivity properties in the presence and absence of bounded disturbances. 

We first analyze the system \eqref{ode} in the absence of disturbances. Consider the candidate Lyapunov function $V(z) = 0.5z^2$, and its derivative along (\ref{ode}) is 
\begin{equation}\label{dV_ode}
	\begin{aligned}
		\dot{V}(z) = -2\alpha V + \beta.
	\end{aligned}
\end{equation}
From (\ref{dV_ode}), it is obtained that system (\ref{ode}) is practically exponentially stable. Moreover, by solving the differential equation (\ref{ode}) and setting the initial condition as $z(t_0)\in\mathbb{R}_+$, its solution is expressed as follows
\begin{equation}\label{z_ode}
	z(t) = \sqrt{\bigg(z^2(t_0) - \frac{\beta}{\alpha}\bigg)e^{-2\alpha t} + \frac{\beta}{\alpha}}.
\end{equation} 
Then, it is clear that $z(t)$ remains strictly positive and monotonically converges to $\sqrt{\beta/\alpha}$ as $t$ tends to infinity. 

Next, consider the disturbed case, where the system is subject to a bounded signal composed of both disturbances and uncertainties. Specifically, we examine the following dynamics
\begin{equation}\label{ode_w}
	\dot{z} = -\alpha z + \frac{\beta}{z} + w(t,z),
\end{equation}
where the disturbance term $w(t,z)$ is continuous in $t$, and locally Lipschitz continuous in $z$ and satisfying $\vert w(t,z)\vert \leq \bar{w}$ with $\bar{w}\in\mathbb{R}_+$. The following results can be obtained.
\begin{lemma}
    Considering the system (\ref{ode_w}), if the initial conditions $z(t_0) > 0$, then $z(t) > 0,\;\forall t\geq t_0$.
\end{lemma}
\begin{proof}
    Based on the existence and uniqueness theorem \cite{khalil2002nonlinear}, for any given initial state $z(t_0)\in\mathbb{R}_+$, there exists a maximum time interval $\mathcal{I}_z(z(t_0)):=[t_0,\tau_{max})$ with $\tau_{max}\in\mathbb{R}_+$ such that $z(t)$ is the unique solution on $\mathcal{I}_z(z(t_0))$. Define $V(z) = 0.5z^2$, and its derivative along (\ref{ode_w}) is 
    \begin{equation}
        \begin{aligned}
            \dot{V}(z) &= -\alpha z^2 + zw + \beta \\
            &\leq -(\alpha - \frac{1}{2\epsilon}) z^2 + \frac{\epsilon}{2}\bar{w}^2 + \beta\\
            &= -2\alpha^{'}V + \delta,
        \end{aligned}
    \end{equation}
    where $\alpha^{'} = (\alpha - \frac{1}{2\epsilon})$, $\delta = \frac{\epsilon}{2}\bar{w}^2 + \beta$ and $\epsilon>0$ is the parameter introduced via Young’s inequality. It is clear that the system (\ref{ode_w}) is input-to-state stable (ISS) \cite{sontag1995characterizations} with respect to $w(t,z)$ on $\mathcal{I}_z(z(t_0))$. 

    \begin{figure}[t]
	\centering
	\includegraphics[width=0.4\textwidth]{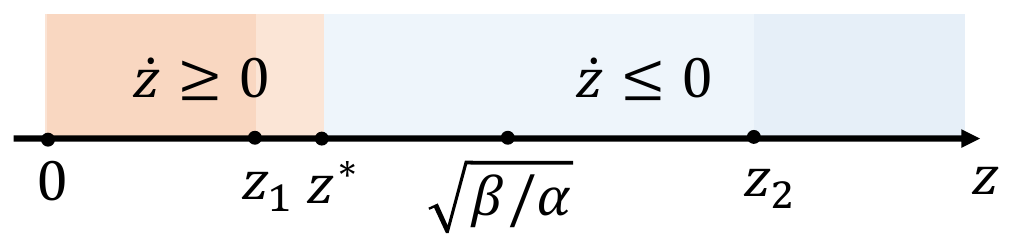}
	\vspace{-0pt}
	\caption{Illustration of the sign of $\dot{z}$ under bounded disturbance. For $(0,z^*]$, $\dot{z}\geq 0$ and $\dot{z}\leq0$ for $[z^*,+\infty)$. Notably, for any $w$ satisfying $\vert w\vert \leq \bar{w}$, it is necessarily guaranteed that $\dot{z}\geq0$ for $(0,z_1]$ and $\dot{z}\leq0$ for $[z_2,+\infty)$.}
	\label{fig_dotz}
    \end{figure}
    
    To establish positivity of $z(t)$, inspired by the analysis in \cite{guo2017simple}, we first consider the worst case, where the disturbance tends to drive $z$ toward zero, i.e., $w(t,z)=-\bar{w}$ 
    \begin{equation}
		\begin{aligned}
		  \dot{z} &\geq -\alpha z + \frac{\beta}{z} - \bar{w} = \frac{-\alpha z^2 - \bar{w}z + \beta}{z}.
		\end{aligned}
    \end{equation}
    Define $s(m) = -\alpha m^2 - \bar{w}m + \beta,\; m\in\mathbb{R}$. The sign of $\dot{z}$ is determined by the quadratic $s(z)$, and its unique positive root of $s(z_1)=0$ is
    \begin{equation}\label{z1}
        z_1 = \frac{\sqrt{\bar{w}^2 + 4\alpha\beta} - \bar{w}}{2\alpha}.
    \end{equation}
    Then, for $z\in(0,z_1]$, we have $\dot{z}\geq 0$, indicating that $z(t)$ is strictly non-decreasing in this region, where the term $\beta/z$ dominates the dynamics. 
    
    Next, consider the opposite scenario where the disturbance maximally assists the increase of $z$, i.e., $w(t,z)=\bar{w}$
    \begin{equation}
		\begin{aligned}
		  \dot{z} &\leq  \frac{-\alpha z^2 + \bar{w}z + \beta}{z}.
		\end{aligned}
    \end{equation}
    The corresponding positive root $z_2 $ is
    \begin{equation}\label{z2}
        z_2 = \frac{\sqrt{\bar{w}^2 + 4\alpha\beta} + \bar{w}}{2\alpha}.
    \end{equation}
    Thus, for $z\in[z_2,\infty)$, $\dot{z} \leq 0$, so $z(t)$ is strictly non-increasing in this region, where the term $-\alpha z$ is dominant. 
    
    As for any $z\in(z_1,z_2)$, the sign of $\dot{z}$ depends on the actual value of $w$. Let $z^*\in\mathbb{R}_+$ denote the equilibrium of the disturbed dynamics. Based on the above analysis, we have $\dot{z}\geq 0$ for $z\in(0,z^*]$ and $\dot{z}\leq 0$ for $z\in[z^*,\infty)$. In particular, $z^*=z_1$ if $w(t,z)=-\bar{w}$, $z^*=z_2$ if $w(t,z)=\bar{w}$ and $z^*=\sqrt{\beta/\alpha}$ if $w(t,z)=0$. The qualitative behavior of $\dot{z}$ under these cases is depicted in Fig. \ref{fig_dotz}.
    
    Therefore, for any $w$ satisfying $\vert w\vert \leq \bar{w}$, it necessarily holds that $\dot{z}\geq 0$ for $z\in(0,z_1]$, and the vector field points inward in the region. According to Nagumo's Theorem \cite{blanchini1999set,brezis1970characterization}, the set $(0,+\infty)$ is forward invariant, i.e., $z(t)\in\mathbb{R}_+$ on $\mathcal{I}_z(z(t_0))$.
    
    Now we show that $\tau_{max}=\infty$. Following the above analysis, $z(t)$ is continuous and strictly positive on $\mathcal{I}_z(z(t_0))$. Therefore, by continuity, the solution cannot reach zero in finite time. Furthermore, since the vector field in (\ref{ode_w}) is locally Lipschitz on $\mathbb{R}_+$, and the system (\ref{ode_w}) is ISS with respect to $w(t,z)$, it is concluded that each solution $z(t)$ starting from $z(t_0)\in\mathbb{R}_+$ is well defined and $z(t)>0$ over $t\in[t_0,\infty)$. This completes the proof.		  
\end{proof}

From (\ref{dV_ode}), (\ref{z_ode}) and Lemma 1, the following conclusions can be drawn
\begin{itemize}
	\item The system (\ref{ode}) is practically exponentially stable, and $z(t)$ exponentially converges to $\lim\limits_{t\rightarrow+\infty}z(t)=\sqrt{\beta/\alpha}$. 
        \item The system (\ref{ode_w}) is ISS with respect to bounded signal $w(t,z)$. Moreover, if $z(t_0) > 0$, $z(t)$ remains strictly positive for all $t\geq t_0$ even in the presence of disturbances.
\end{itemize}
Motivated by these insights, we extend the use of the reciprocal of the safe function $h(\boldsymbol{x})$ within the framework of CBF to improve robustness against disturbances.
\subsection{Reciprocal resistance-based barrier function}
Consider the nonlinear system (\ref{sys0}) without $\boldsymbol{u}$
\begin{equation}
	\begin{aligned}
		&\dot{\boldsymbol{x}} = \boldsymbol{f}(\boldsymbol{x}),
	\end{aligned}
	\label{sys_0}
\end{equation}
Denote $\mathcal{I}_x(\boldsymbol{x}(t_0)):=[t_0,\tau_{max})$ as the maximum time interval starting from $\boldsymbol{x}(t_0)$, then $\boldsymbol{x}(t)$ is the unique solution on $\mathcal{I}(\boldsymbol{x}(t_0))$. Inspired by the positivity of $z(t)$, the following reciprocal resistance-based barrier function is defined.
\begin{definition}
	Considering the system (\ref{sys_0}), a continuously differentiable function $h: \mathbb{R}^n\rightarrow\mathbb{R}$ is a reciprocal resistance-based barrier function if there exist extended class $\mathcal{K}$ functions $\alpha$ and $\beta$ subject to
	\begin{equation}
		\begin{aligned}
			\mathcal{L}_{\boldsymbol{f}}h(\boldsymbol{x}) + \alpha(h(\boldsymbol{x})) - \beta(1/h(\boldsymbol{x})) \geq 0,
		\end{aligned}
		\label{rrbf}
	\end{equation}
    for all $\boldsymbol{x}\in{\rm Int}\mathbb{C}$, where ${\rm Int}\mathbb{C}:=\{ \boldsymbol{x}\in \mathbb{X}:h(\boldsymbol{x}) > 0\}$ is the non-empty interior of the set $\mathbb{C}$. 
\end{definition}

Before illustrating the relationship of RRBF with forward invariance of set $\mathbb{C}$, we first define the following auxiliary set $\mathbb{S}$
\begin{equation}\label{set_S}
	\begin{aligned}
		\mathbb{S} &= \{ \boldsymbol{x}\in \mathbb{X}:h(\boldsymbol{x})\geq h_s\},\\
		\partial \mathbb{S}&=\{ \boldsymbol{x}\in \mathbb{X}:h(\boldsymbol{x})=h_s\},\\
		{\rm Int}\mathbb{S}&=\{ \boldsymbol{x}\in \mathbb{X}:h(\boldsymbol{x})>h_s\},
	\end{aligned}
\end{equation}
where $\partial \mathbb{S}$, ${\rm Int}\mathbb{S}$ are the boundary and the interior of the set $\mathbb{S}$, respectively. In (\ref{set_S}), $h_s\in\mathbb{R}_+$ is the solution of $\alpha(h)-\beta(1/h) = 0$. The following lemma demonstrates the existence and uniqueness of $h_s$.
\begin{lemma}
    For any extended class $\mathcal{K}$ functions $\alpha(s)$ and $\beta(s)$ with $s\in\mathbb{R}_+$, there exists a unique $s^*\in\mathbb{R}_+$ such that $\alpha(s^*)-\beta(1/s^*)=0$.
\end{lemma}
\begin{proof}
    Defining $f(s) = \alpha(s)-\beta(1/s)$ on $s\in\mathbb{R}_+$, then we have 
    \begin{equation}
        \begin{aligned}
            \lim_{s\rightarrow 0} f(s) = -\infty,\;        \lim_{s\rightarrow \infty} f(s) = +\infty.
        \end{aligned}
    \end{equation}
    Since $\alpha(s)$ is a strictly increasing function and $\beta(1/s)$ is a strictly decreasing function, it follows that $f(s)$ is a strictly increasing continuous function on $s\in\mathbb{R}_+$. Based on the Intermediate Value Theorem \cite{bartle2000introduction}, there exists a unique point $s^*$ such that $\alpha(s^*) = \beta(1/s^*)$. This completes the proof.
\end{proof}

Then, it is able to conclude the following Theorem.
\begin{theorem}
	Consider the system (\ref{sys_0}) and the set $\boldsymbol{\mathbb{S}}$ given in (\ref{set_S}). With a valid RRBF $h: \mathbb{R}^n\rightarrow\mathbb{R}$, if the initial conditions of $h(\boldsymbol{x})\in {\rm Int \mathbb{C}}$, then ${\rm Int \mathbb{C}}$ is forward invariant.
\end{theorem}
\begin{proof}
    First, consider the case where the initial condition satisfies $\boldsymbol{x}(t_0) \in \mathbb{S}$. For any state $\boldsymbol{x} \in \partial \mathbb{S}$ (i.e., on the boundary of $\mathbb{S}$), it follows from the RRBF condition (\ref{rrbf}) and Lemma 2 that
    \begin{equation}
		\begin{aligned}
			\dot{h}(\boldsymbol{x}) &\geq -\alpha( h(\boldsymbol{x})) + \beta(1/h(\boldsymbol{x})) = 0,\; \forall \boldsymbol{x}\in\partial\mathbb{S}.
		\end{aligned}
    \end{equation}
    According to Nagumo's theorem \cite{blanchini1999set,brezis1970characterization}, the set $\mathbb{S}$ is forward invariant.

    Next, for the case where the initial state lies in the region $\mathrm{Int}\mathbb{C} \setminus \mathbb{S}$, define the following Lyapunov-like function
    \begin{equation}
        V_C(\boldsymbol{x}) = \left\{
        \begin{aligned}
            &0,\;\; {\rm if}\; \boldsymbol{x}\in\mathbb{S}\\
            &h_s - h(\boldsymbol{x}),\;\; {\rm if}\; \boldsymbol{x}\in{\rm Int}\mathbb{C}\setminus\mathbb{S}
        \end{aligned}
        \right.
    \end{equation}
    It is clear that $V_C(\boldsymbol{x})$ is positive $\forall \boldsymbol{x}\in{\rm Int}\mathbb{C}\setminus\mathbb{S}$ and 
    \begin{equation}
        \dot{V}_C(\boldsymbol{x}) \leq \alpha( h(\boldsymbol{x})) - \beta(1/h(\boldsymbol{x})) < 0, \; \forall \boldsymbol{x}\in{\rm Int}\mathbb{C}\setminus\mathbb{S}.
    \end{equation}
    This implies that $V_C$ decreases along system trajectories outside $\mathbb{S}$, and thus the trajectory asymptotically converges to $\mathbb{S}$. Since $\mathbb{S}$ itself is forward invariant, it concludes that any trajectory starting in $\mathrm{Int}\mathbb{C}$ will remain in $\mathrm{Int}\mathbb{C}$. This completes the proof. 
\end{proof}

In the following, we demonstrate the enhanced robustness against disturbances $\boldsymbol{d}$ of RRBF.
\begin{figure}[t]
	\centering
	\includegraphics[width=0.4\textwidth]{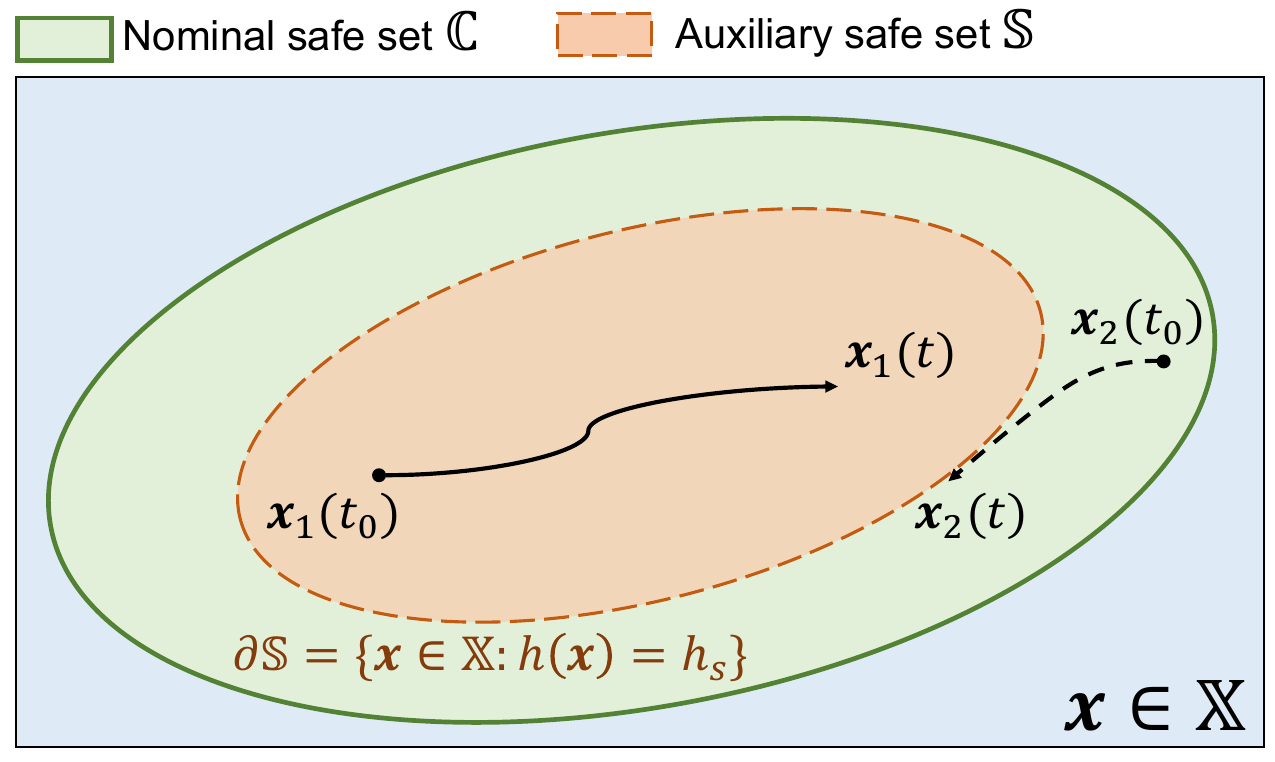}
	\vspace{-0pt}
	\caption{Illustrative set invariance of RRBF. Two different initial conditions are presented, i.e., $\boldsymbol{x}_1(t_0)\in\mathbb{S}$ and $\boldsymbol{x}_2(t_0)\in{\rm Int}\mathbb{C}/\mathbb{S}$. The trajectory starts from $\boldsymbol{x}_1(t_0)$ will be within the set $\mathbb{S}$, and the trajectory starts from $\boldsymbol{x}_2(t_0)$ will converge to the boundary $\partial \mathbb{S}$.}
	\label{fig_RRBF}
\end{figure}
\subsection{Set invariance with bounded disturbance}
Consider the nonlinear system (\ref{sys0}) with $\boldsymbol{d}(t,\boldsymbol{x})$.
\begin{equation}
	\begin{aligned}
		&\dot{\boldsymbol{x}} = \boldsymbol{f}(\boldsymbol{x}) + \boldsymbol{d}(t,\boldsymbol{x}).
	\end{aligned}
	\label{sys_d}
\end{equation}
We will show that the forward invariance of the system (\ref{sys_d}) with respect to $\mathrm{Int}\mathbb{C}$ is still maintained with a valid RRBF defined in (\ref{rrbf}).
\begin{theorem}
	Consider the system (\ref{sys_d}). With a valid RRBF $h: \mathbb{R}^n\rightarrow\mathbb{R}$, if the initial conditions of $h(\boldsymbol{x})\in \mathrm{Int}\mathbb{C}$, then $\mathrm{Int}\mathbb{C}$ is forward invariant.
\end{theorem}
\begin{proof}
	Considering (\ref{rrbf}) and (\ref{sys_d}), we have 
	\begin{equation}
		\begin{aligned}
			\dot{h}(\boldsymbol{x}) &\geq -\alpha (h(\boldsymbol{x})) + \beta({1/h(\boldsymbol{x})}) + \frac{\partial h(\boldsymbol{x})}{\partial \boldsymbol{x}}\boldsymbol{d}.
		\end{aligned}
	\end{equation}
	Since the $h(\boldsymbol{x})$ is a continuously differentiable function, then $\frac{\partial h(\boldsymbol{x})}{\partial \boldsymbol{x}}$ is continuous and bounded for $\boldsymbol{x} \in \mathbb{X}$. It is obtained that
	 \begin{equation}
	 	\begin{aligned}
	 		\dot{h}(\boldsymbol{x}) &\geq -\alpha (h(\boldsymbol{x})) + \beta({1/h(\boldsymbol{x})}) - \bigg \vert \frac{\partial h(\boldsymbol{x})}{\partial \boldsymbol{x}}\boldsymbol{d} \bigg \vert,\\
            &\geq -\alpha (h(\boldsymbol{x})) + \beta({1/h(\boldsymbol{x})}) - D,
	 	\end{aligned}
	 \end{equation}
	 where $D\in\mathbb{R}_+$ is the upper bound of $ \frac{\partial h(\boldsymbol{x})}{\partial \boldsymbol{x}}\boldsymbol{d},\;\forall \boldsymbol{x}\in\mathbb{X},\;\forall \boldsymbol{d}\in\mathbb{D}$. Following a similar analysis provided in Lemma 2, there exists a unique $h_b\in\mathbb{R}_+$ such that $\alpha(h_b)+D = \beta(1/h_b)$.
     
     Furthermore, it is obtained that $\dot{h}(\boldsymbol{x}) \geq 0$ with $h(\boldsymbol{x})\in(0,h_b]$. Following the positivity analysis in Lemma 1, it is concluded that $h(\boldsymbol{x})$ is positive if $h(\boldsymbol{x}(t_0)) \in \mathrm{Int}\mathbb{C}$ even in the presence of disturbance $\boldsymbol{d}(t,\boldsymbol{x})$. This completes the proof.		  
\end{proof}
From the above analysis, the proposed RRBF ensures a critical safety guarantee in the presence of any bounded disturbance $\boldsymbol{d}(t,\boldsymbol{x})$ without requiring additional prior information.
\subsection{Relationships of RRBF with RBF and ZBF}
Before proceeding with a detailed comparison, we briefly revisit the definitions of RBFs and ZBFs.
\begin{definition}
    Considering the system (\ref{sys_0}), a continuously differentiable function $B: {\rm Int}\mathbb{C}\rightarrow\mathbb{R}$ is a reciprocal barrier function for a continuously differentiable function $h: \mathbb{R}^n\rightarrow \mathbb{R}$ if there exist class $\mathcal{K}$ functions $\alpha_1$, $\alpha_2$ and $\alpha_3$ such that
	\begin{equation}\label{rbf}
		\begin{aligned}
            \frac{1}{\alpha_1(h(\boldsymbol{x}))} &\leq B(\boldsymbol{x}) \leq \frac{1}{\alpha_2(h(\boldsymbol{x}))},\\
		  \mathcal{L}_{\boldsymbol{f}}B(\boldsymbol{x}) &\leq \alpha_3 (h(\boldsymbol{x})),
		\end{aligned}
	\end{equation} 
    for all $\boldsymbol{x}\in{\rm Int}\mathbb{C}$.
\end{definition}
\begin{definition}
    Considering the system (\ref{sys_0}), a continuously differentiable function $h: \mathbb{R}^n\rightarrow \mathbb{R}$ is a zeroing barrier function, if there exists an extended class $\mathcal{K}$ function $\alpha_e$ such that
	\begin{equation}\label{zbf}
		\begin{aligned}
		  \mathcal{L}_{\boldsymbol{f}}h(\boldsymbol{x}) &\geq -\alpha_e(h(\boldsymbol{x}))
		\end{aligned}
	\end{equation} 
    for all $\boldsymbol{x}\in\mathbb{X}$.
\end{definition}

Based on the definitions above, we summarize their key differences among ZBF, RBF and RRBF:
\begin{itemize}
    \item [(1).] \underline{\textit{Forward invariance}}. ZBF ensures forward invariance for the entire set $\mathbb{C}$, while RBF and RRBF guarantee forward invariance for ${\rm Int}\mathbb{C}$. To be specific, RRBF first enforces forward invariance for the subset $\mathbb{S}$, and then provides a buffer zone near the $\partial\mathbb{C}$. Trajectories initialized in this buffer will converge toward the boundary $\partial\mathbb{S}$, as illustrated in Fig. \ref{fig_RRBF}. 
    \item [(2).] \underline{\textit{Robustness}}. As $\boldsymbol{x}(t)$ approaches $\partial\mathbb{C}$, both RBF and ZBF experience diminishing growth rates, making them highly sensitive to disturbances/uncertainties $\boldsymbol{d}(t,\boldsymbol{x})$. In contrast, RRBF exhibits superior robustness to bounded disturbances, as analyzed in Theorem 2, by maintaining a non-zero domination effect near $\partial\mathbb{C}$. 
    \item [(3).] \underline{\textit{Singularity issue}}. It should be noted that the proposed RRBF does not explicitly account for the boundary $\partial \mathbb{C}$, i.e., $h(\boldsymbol{x}) = 0$. This demonstrates that improving robustness against disturbances comes at the cost of introducing singularity concerns. A practical remedy involves replacing the term $1/h(\boldsymbol{x})$ with $1/[h(\boldsymbol{x}) + \sigma]$, where $\sigma$ is a small positive constant. This modification mitigates singularities near the boundary and is formalized in Lemma 3 in the appendix. It should be noted that the choice of $\sigma$ should be guided by the known upper bound of $\boldsymbol{d}$.
\end{itemize}
\section{Reciprocal resistance-based control barrier function}
In this section, we first introduce the definition of the RRCBF for system (\ref{sys}). The concept is then extended to cases where the input relative degree of $h(\boldsymbol{x})$ with respect to system (\ref{sys}) is greater than one. Finally, we integrate disturbance estimation and compensation techniques with RRCBF to further improve the control performance with complex uncertainties.
\subsection{Reciprocal resistance-based control barrier function}
Consider the disturbed affine nonlinear system in (\ref{sys}).
\begin{definition}
	Considering the system (\ref{sys}), a continuously differentiable function $h: \mathbb{R}^n\rightarrow\mathbb{R}$ is a reciprocal resistance-based control barrier function if there exist extended class $\mathcal{K}$ functions $\alpha$ and $\beta$ subject to
	\begin{equation}
		\begin{aligned}
			\sup_{\boldsymbol{u}\in \mathbb{R}^m} \bigg[ \mathcal{L}_{\boldsymbol{f}}h(\boldsymbol{x}) + \mathcal{L}_{\boldsymbol{g}}h(\boldsymbol{x})\boldsymbol{u} + \alpha (h(\boldsymbol{x})) - \beta({1/h(\boldsymbol{x})}) \bigg] \geq 0
		\end{aligned}
		\label{RRcbf}
	\end{equation}
    for all $\boldsymbol{x}\in{\rm Int}\mathbb{C}$.
\end{definition}
Given an RRCBF $h(\boldsymbol{x})$, define the set of control input $\boldsymbol{u}$.
\begin{equation}\label{K-RRCBF}
    \begin{aligned}
        \mathcal{K}_{\rm RRCBF} = \{ \boldsymbol{u}\in\mathbb{R}^m: \mathcal{L}_{\boldsymbol{f}}h(\boldsymbol{x}) + \mathcal{L}_{\boldsymbol{g}}h(\boldsymbol{x})\boldsymbol{u} + \alpha (h(\boldsymbol{x}))\\
        - \beta({1/h(\boldsymbol{x})}) \geq 0 \}.
    \end{aligned}
\end{equation}
The following results guarantee the forward invariance of ${\rm Int}\mathbb{C}$.
\begin{corollary}
	Consider the system (\ref{sys}) with bounded disturbances $\boldsymbol{d}$. Giving a valid RRCBF $h: \mathbb{R}^n\rightarrow\mathbb{R}$, if the initial conditions of $h(\boldsymbol{x}(t_0))\in{\rm Int}\mathbb{C}$, then any Lipschitz continuous controllers belong to the set $\mathcal{K}_{\rm RRCBF}$ will render the set ${\rm Int}\mathbb{C}$ forward invariant.
\end{corollary}
\begin{proof}
	Substituting the control from $\mathcal{K}_{\rm RRCBF}$ into (\ref{sys}), one yields
	\begin{equation}
		\dot{h}(\boldsymbol{x}) \geq -\alpha(h(\boldsymbol{x})) + \beta(1/h(\boldsymbol{x})) + \mathcal{L}_{\boldsymbol{g}}h(\boldsymbol{x})\boldsymbol{d},
	\end{equation}
	Since $h(\boldsymbol{x})$ is continuously differentiable and $g(\boldsymbol{x})$ is Lipschitz continuous, then $\mathcal{L}_{\boldsymbol{g}}h(\boldsymbol{x})$ is bounded for $\boldsymbol{x}\in\mathbb{X}$. Based on the results provided in Theorems 1 and 2, it is obtained that ${\rm Int}\mathbb{C}$ is forward invariant. This completes the proof.
\end{proof}
\subsection{High-order reciprocal resistance-based control barrier function}
In the following, the case when the $h(\boldsymbol{x})$ has a relative degree $r\geq 2$ is discussed. Here, we assume that $\mathcal{L}_{\boldsymbol{g}}\mathcal{L}^k_{\boldsymbol{f}}h(\boldsymbol{x}) = 0$ with $k\in\mathbb{N}_{0:r-2}$ and $\mathcal{L}_{\boldsymbol{g}}\mathcal{L}^{r-1}_{\boldsymbol{f}}h(\boldsymbol{x}) \neq 0$. Based on \cite{xiao2021high,tan2021high}, we first define a sequence of functions $\psi_i:\mathbb{R}^n\rightarrow\mathbb{R}, i\in\mathbb{N}_{1:r}$, as follows
\begin{equation}\label{psi_state}
	\begin{aligned}
		\psi_0(\boldsymbol{x}) &= h(\boldsymbol{x}),\;\\
		\psi_i(\boldsymbol{x}) &= \dot{\psi}_{i-1}(\boldsymbol{x}) + \alpha_i(\psi_{i-1}(\boldsymbol{x})), i\in\mathbb{N}_{1:r},
	\end{aligned}
\end{equation}
where $\alpha_i$ are extended class $\mathcal{K}$ functions to be selected. We further define the following 0-superlevel sets $\mathbb{S}_{i},\;i\in\mathbb{N}_{0:r-1}$
\begin{equation}
    \begin{aligned}
		\mathbb{S}_{0}&=\{\boldsymbol{x}\in\mathbb{R}^n: \psi_{0}(\boldsymbol{x})\geq 0\},\\
		\mathbb{S}_{i}&=\{\boldsymbol{x}\in\mathbb{R}^n: \psi_{i-2}(\boldsymbol{x})\geq 0\},\; i\in\mathbb{N}_{2:r-2},\\
		\mathbb{S}_{r-1}&=\{\boldsymbol{x}\in\mathbb{R}^n: \psi_{r-1}(\boldsymbol{x})> 0 \}.
    \end{aligned}
\end{equation}
\begin{definition}
	Consider the nonlinear system (\ref{sys}) and $\psi_i(\boldsymbol{x}),\;i\in\mathbb{N}_{1:r}$. A continuously differentiable function $h: \mathbb{R}^n\rightarrow\mathbb{R}$ is a HO-RRCBF with input relative degree $r$ for system (\ref{sys}), if there exist extended class $\mathcal{K}$ functions $\alpha_i,\;i\in\mathbb{N}_{1:r}$ and $\beta$, subject to
	\begin{equation}
		\begin{aligned}
			\sup_{\boldsymbol{u}\in \mathbb{R}^m} \bigg[ \mathcal{L}_{\boldsymbol{f}}^rh(\boldsymbol{x}) + \mathcal{L}_{\boldsymbol{g}}\mathcal{L}_{\boldsymbol{f}}^{r-1}h(\boldsymbol{x})\boldsymbol{u} + \alpha_r(\psi_{r-1}(\boldsymbol{x}))\\ - \beta(1/\psi_{r-1}(\boldsymbol{x}))
		  + O(h(\boldsymbol{x})) \bigg] \geq 0,
		\end{aligned}
		\label{hoRRcbf}
	\end{equation}
	for all $\boldsymbol{x}\in\bar{\mathbb{S}}:=\mathbb{S}_0\cap\mathbb{S}_1\cap\ldots\cap\mathbb{S}_{r-1}$, where $O(h(\boldsymbol{x})) = \sum_{i=1}^{r-1}\mathcal{L}_{\boldsymbol{f}}^i(\alpha_{r-1}\circ\psi_{r-i-1}(\boldsymbol{x}))$. 
\end{definition} 

If there exists a valid HO-RRCBF, the forward invariance of set $\bar{\mathbb{S}}$ is summarized as the following Theorem.
\begin{theorem}
    Consider the system (\ref{sys}) with bounded disturbances $\boldsymbol{d}$. Giving a valid HO-RRCBF $h: \mathbb{R}^n\rightarrow\mathbb{R}$, if the initial conditions of $\psi_{i}(\boldsymbol{x}(t_0)),\;i\in\mathbb{N}_{0:r-1}$ are positive, then any Lipschitz continuous controllers belong to the set $\mathcal{K}_{\rm HO-RRCBF}:=\{\boldsymbol{u}\in\mathbb{R}^m: \mathcal{L}_{\boldsymbol{f}}^rh(\boldsymbol{x}) + \mathcal{L}_{\boldsymbol{g}}\mathcal{L}_{\boldsymbol{f}}^{r-1}h(\boldsymbol{x})\boldsymbol{u} + \alpha_r(\psi_{r-1}(\boldsymbol{x})) - \beta(1/\psi_{r-1}(\boldsymbol{x})) + O(h(\boldsymbol{x}))\geq 0 \}$ will render the set $\bar{\mathbb{S}}$ forward invariant.
\end{theorem}
\begin{proof}
    First, since there exists a Lipschitz continuous controller $\boldsymbol{u}(\boldsymbol{x})\in \mathcal{K}_{\rm HO-RRCBF}$, then by substituting $\psi_r(\boldsymbol{x},\boldsymbol{u})$ into the dynamics of $\psi_{r-1}(\boldsymbol{x})$ in (\ref{psi_state}), we have
    \begin{equation}
		\dot{\psi}_{r-1}(\boldsymbol{x}) \geq -\alpha_r(\psi_{r-1}(\boldsymbol{x})) + \beta({1/\psi_{r-1}(\boldsymbol{x})}) + \mathcal{L}_{\boldsymbol{g}}\mathcal{L}_{\boldsymbol{f}}^{r-1}h(\boldsymbol{x})\boldsymbol{d}.
    \end{equation}
    Following the results of Theorems 1 and 2, it is obtained that $\psi_{r-1}(\boldsymbol{x})$ is positive, i.e., the set $\mathbb{S}_{r-1}$ is forward invariant. 
		
    Then, the forward invariance of $\mathbb{S}_{r-2}$ can also be guaranteed for that of $\mathbb{S}_{r-1}$ and $\psi_{r-2}(\boldsymbol{x}(t_0))\geq 0$. Therefore, following a recursive analysis from the forward invariance of $\mathbb{S}_{r-1}$ to that of $\mathbb{S}_0$, the forward invariant property of set $\bar{\mathbb{S}}$ can be achieved such that the strict safety constraint is guaranteed. This completes the proof.
\end{proof}
\begin{figure*}[!t]
	\centering
	\subfloat[]{
		\includegraphics[width=0.31\textwidth]{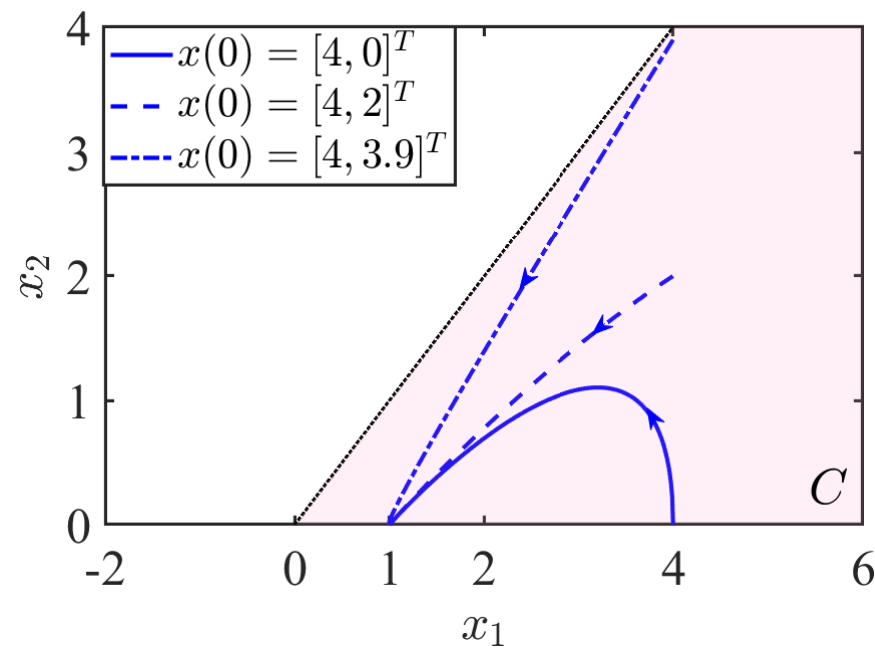}}
	\subfloat[]{
		\includegraphics[width=0.31\textwidth]{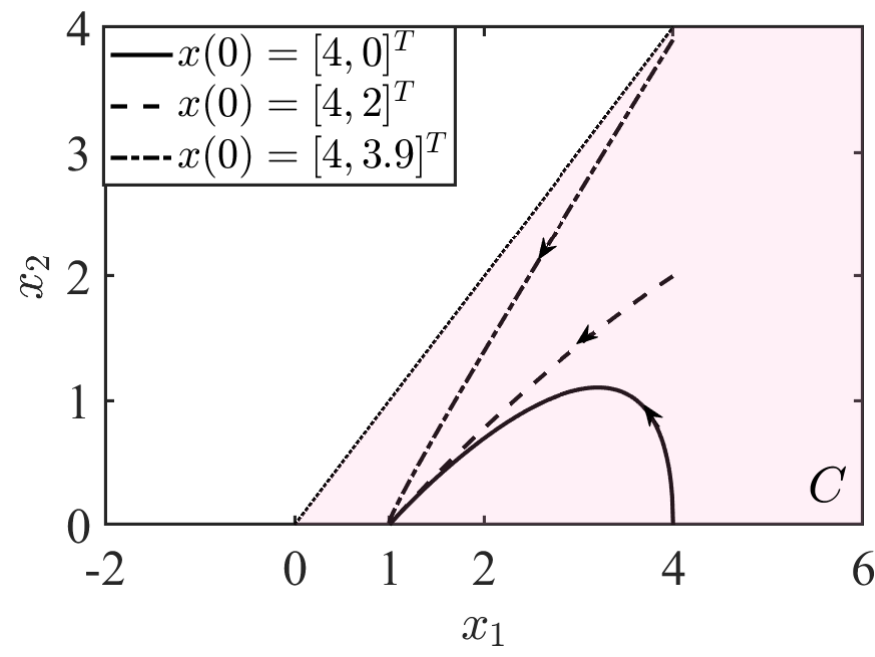}}
        \subfloat[]{
		\includegraphics[width=0.31\textwidth]{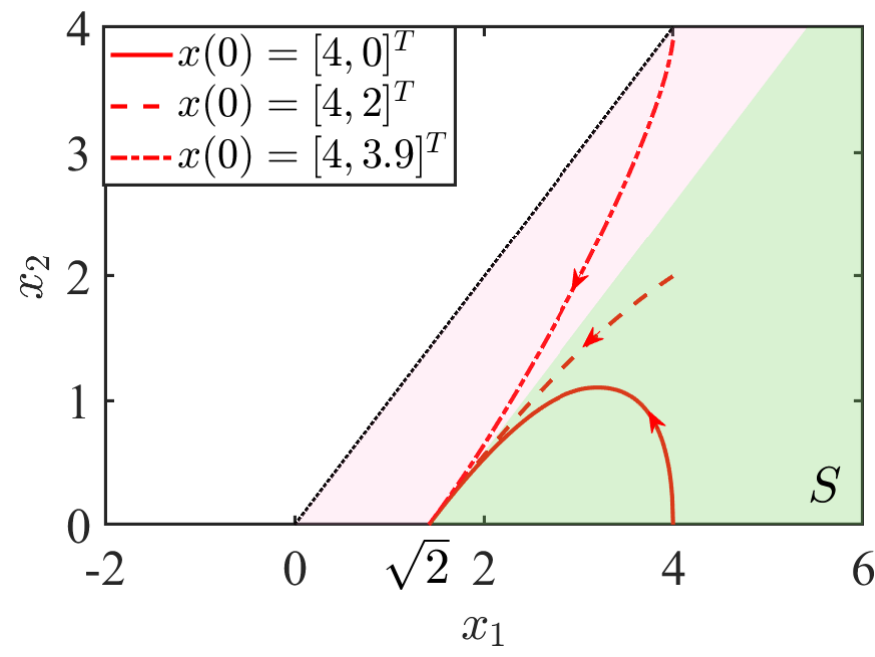}}
	\caption{State trajectories of system \eqref{exp} in the phase plane under different CBF designs with $w=0$: (a) ZCBF with different initial conditions, (b) RCBF with different initial conditions, (c) RRCBF with different initial conditions. }
	\label{results1-1}
\end{figure*} 
\section{RRCBF with disturbance estimate}
The key advantage of the proposed RRCBF lies in the inclusion of the terms $\beta/h(\boldsymbol{x})$ and $\beta/\psi_{r-1}(\boldsymbol{x})$, which enhance robustness against $\boldsymbol{d}$. However, this may introduce conservatism, especially under complex, time-varying disturbances. To address this, based on disturbance estimation and compensation techniques \cite{DasarXiv,sun2024safety,chen2023robust,takano2020robust,WangarXiv,wang2024tac}, a disturbance observer-based RRCBF (DO-RRCBF) is proposed that explicitly incorporates the disturbance estimate $\hat{\boldsymbol{d}}$.
\subsection{DO-RRCBF design}
We assume the availability of well-designed disturbance observers capable of accurately estimating the disturbance $\hat{\boldsymbol{d}}$.
\begin{definition}
	Considering the system (\ref{sys}), a continuously differentiable function $h: \mathbb{R}^n\rightarrow\mathbb{R}$ is a DO-RRCBF, if there exists extended class $\mathcal{K}$ functions $\alpha$ and $\beta$ subject to
	\begin{equation}
		\begin{aligned}
			\sup_{\boldsymbol{u}\in \mathbb{R}^m} \bigg[ \mathcal{L}_{\boldsymbol{f}}h(\boldsymbol{x}) + \mathcal{L}_{\boldsymbol{g}}h(\boldsymbol{x})(\boldsymbol{u} + \hat{\boldsymbol{d}}) + \alpha (h(\boldsymbol{x})) - {\beta}({1/h(\boldsymbol{x})}) \bigg] \geq 0
		\end{aligned}
		\label{do-RRcbf}
	\end{equation}
    for all $\boldsymbol{x}\in {\rm Int}\mathbb{C}$.
\end{definition}
\begin{corollary}
	Consider the system (\ref{sys}) with bounded external disturbances $\boldsymbol{d}$. Giving a valid DO-RRCBF $h: \mathbb{R}^n\rightarrow\mathbb{R}$ , if the initial conditions of $h(\boldsymbol{x}(0))\in{\rm Int \mathbb{C}}$, then any Lipschitz continuous controllers belong to the set $\mathcal{K}_{\rm DO-RRCBF}:= \{\boldsymbol{u}\in\mathbb{R}^m: \mathcal{L}_{\boldsymbol{f}}h(\boldsymbol{x}) + \mathcal{L}_{\boldsymbol{g}}h(\boldsymbol{x})(\boldsymbol{u} + \hat{\boldsymbol{d}}) + \alpha (h(\boldsymbol{x})) - {\beta}({1/h(\boldsymbol{x})}) \geq 0\}$ will render the set ${\rm Int \mathbb{C}}$ forward invariant.
\end{corollary}

The proof of Corollary 2 can be derived by following the results of Theorems 1 and 2 and Corollary 1. The key distinction is the inclusion of the term $\beta(1/h(\boldsymbol{x}))$, which serves to dominate the estimation error of $\hat{\boldsymbol{d}}$. Furthermore, it is worth noting that the definition of DO-RRCBF can be naturally extended to cases where the input relative degree satisfies $r\geq 2$. 

\subsection{Comparison of DO-RRCBF and disturbance observer-based CBF}
In the work of \cite{Das2023,DasarXiv}, the disturbance observer-based CBF (DO-CBF) is proposed by online compensating the disturbance estimate with the help of the following assumption.
\begin{assumption}\label{assmp1}
    The disturbance $\boldsymbol{d}\in\mathbb{D}$ is assumed to be smooth with respect to $t$ and satisfy $\Vert \boldsymbol{d} \Vert \leq \delta_0$ and $\Vert \dot{\boldsymbol{d}} \Vert \leq \delta_1$, where $\delta_0$ and $\delta_1$ are known positive constants. 
\end{assumption}

Under this assumption, the DO-CBF is defined as follows.
\begin{definition}
	Considering the system (\ref{sys}) under Assumption 1, a continuously differentiable function $h: \mathbb{R}^n\rightarrow\mathbb{R}$ is a DO-CBF, if there exists an extended class $\mathcal{K}$ function $\alpha_e(\cdot)$ subject to
	\begin{equation}
		\begin{aligned}
			\sup_{\boldsymbol{u}\in \mathbb{R}^m}\bigg[ \mathcal{L}_{\boldsymbol{f}}h(\boldsymbol{x}) + \mathcal{L}_{\boldsymbol{g}}h(\boldsymbol{x})(\boldsymbol{u} + \hat{\boldsymbol{d}}) - \Vert \mathcal{L}_{\boldsymbol{g}}h(\boldsymbol{x})\Vert{\epsilon}_d + \alpha_e(h(\boldsymbol{x})) \bigg] \geq 0
		\end{aligned}
		\label{do_cbf}
	\end{equation}
    for all $\boldsymbol{x}\in \mathbb{X}$.
\end{definition}

In (\ref{do_cbf}), ${\epsilon}_d(t)$ denotes the bound on the estimation error, which is used to ensure strict safety enforcement by conservatively dominating the observer's estimation error. Generally, the estimation error bound $\epsilon_d(t)$ is a monotonically decreasing function converging to a small positive value \cite{DasarXiv,WangarXiv,AnilarXiv}. However, the construction of ${\epsilon}_d(t)$ usually requires the prior information of $\boldsymbol{d}$ like $\delta_0$ and $\delta_1$.

In contrast, the proposed DO-RRCBF replaces the estimation error term $\Vert \mathcal{L}_{\boldsymbol{g}}h(\boldsymbol{x})\Vert{\epsilon}_d$ with the reciprocal resistance term $\beta(1/h(\boldsymbol{x}))$, thereby reducing reliance on Assumption 1 and simplifying practical implementation. This benefit will be further demonstrated through simulations in the following section.
\section{Simulation results}
In this section, two simulation examples are presented. The first example considers a second-order linear system with disturbances, aiming at illustrating the forward invariant sets of ZBF, RBF and RRBF in the phase plane. The second example is the adaptive cruise control problem as outlined in \cite{Ames2017} to demonstrate the effectiveness of the proposed approach with high control relative degree. 

\subsection{An illustrative example}
Consider the following linear system employed from \cite{alan2021safe}.
\begin{equation}\label{exp}
	\begin{aligned}
		\dot{x}_1 &= -x_2,\\
        \dot{x}_2 &= u + w(t),
	\end{aligned}
\end{equation}
where $\boldsymbol{x}=[x_1,x_2]^T$ are the states, $u$ is the control input and $w(t)$ is the external disturbance. The desired safety specification is set as $h(\boldsymbol{x}) = x_1 - x_2$. 

To evaluate the set invariance with different CBFs, we consider three CBF designs: the ZCBF, the RCBF with $B(\boldsymbol{x}) = 1/h(\boldsymbol{x})$ and the proposed RRCBF. In all cases, the class $\mathcal{K}$ and extended class $\mathcal{K}$ functions are set as linear (proportional) functions. Define the nominal controller as $u_0 = x_1 - 2x_2 - 1$, yielding a desired equilibrium at $\boldsymbol{x}_e=[1,0]^T$. A quadratic programming-based safety filter technique is used to mitigate both the stability and safety control tasks \cite{Ames2019control}. The control parameters are denoted as $\alpha = 1$, $\beta = 2$ for different CBFs. 

The simulation results are shown in Figs. \ref{results1-1} and \ref{results1-2}. Fig. \ref{results1-1} (a)–(c) demonstrate the system behavior in the absence of disturbances. As illustrated, the proposed RRCBF ensures forward invariance of ${\rm Int}\mathbb{C}$, i.e., the set $\mathbb{S}$ remains invariant for initial conditions 
$\boldsymbol{x}(t_0)\in\mathbb{S}$, and the state converges asymptotically toward the boundary $\partial\mathbb{S}$ when $\boldsymbol{x}(t_0)\in{\rm Int}\mathbb{C}\setminus\mathbb{S}$. These results are in full agreement with the theoretical guarantees established in Theorem 1.
\begin{figure}[!t]
	\centering
		\includegraphics[width=0.45\textwidth]{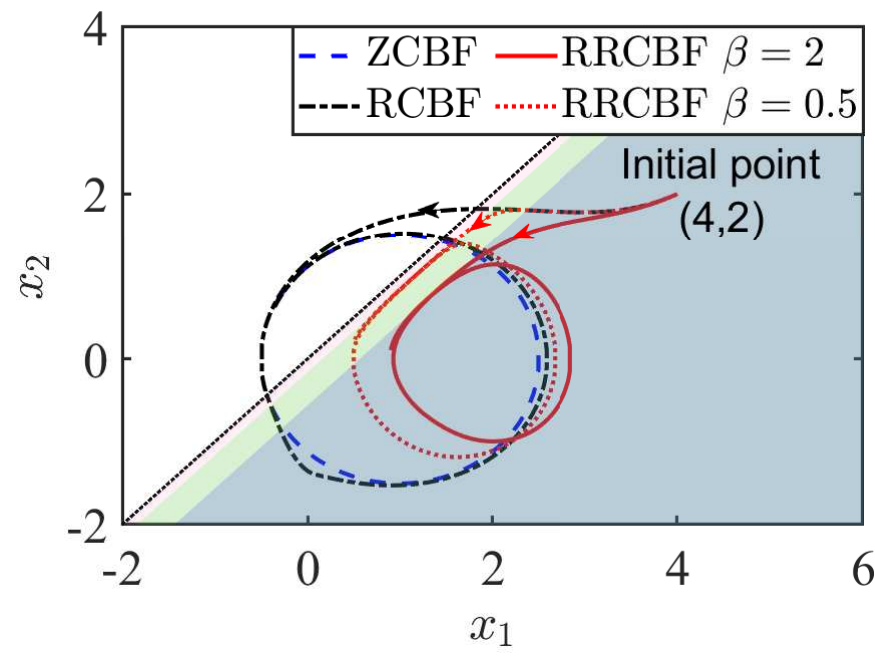}
	\caption{State trajectories of system \eqref{exp} in the phase plane under different CBF designs subject to bounded disturbance. }
	\label{results1-2}
\end{figure} 

Fig. \ref{results1-2} shows the responses under $w(t)=3\sin(t)$. The ZCBF and RCBF cannot maintain set invariance, allowing trajectories to exit the safe set. In contrast, the RRBF preserves safety by preventing boundary violations. Notably, the conservatism of the RRCBF reflected in the size of the robust safe set $\mathbb{S}$ can be tuned by adjusting the parameter $\beta$, i.e., reducing $\beta$ enlarges $\mathbb{S}$.
\subsection{Adaptive cruise control}
The dynamics of the ACC system are described as follows
\begin{equation}
	\begin{aligned}
		\dot{v}_l &= a_l,\\
		\dot{v}_e &= -\frac{1}{m}(f_0 + f_1v_e + f_2v_e^2) + \frac{1}{m}u + w,\\
		\dot{D} &= v_l - v_e,
	\end{aligned}
\end{equation} 
where $v_l$ and $v_e$ are velocities of the leading and ego vehicles; $D$ is the following distance between the ego and leader vehicles; $a_l$ is the acceleration or deceleration of the leader, and $u$ is the wheel force acting as the control input of the ego vehicle; $m$ is the mass of ego vehicle; $f_0,f_1$ and $f_2$ are constant parameters to describe the aerodynamic drag; $w$ is the lumped disturbance. Two control objectives are considered.
\begin{itemize}
	\item \textit{Speed control}. The ego vehicle needs to cruise at the desired speed $v_d$. In order to achieve this, by defining the speed error as $e_v = v_e - v_d$, the following nominal controller is considered.
	\begin{equation}
		u_0 = -m\bigg[ke_v - \frac{1}{m}(f_0 + f_1v_e + f_2v_e^2) + \hat{w}\bigg],
	\end{equation}
	where $k\in\mathbb{R}_+$ is the control parameter to be designed and $\hat{w}$ is the estimate with using disturbance observer proposed in \cite{chen2000nonlinear}.
	\item \textit{Space control}. The safe space maintenance between the ego vehicle and the leading vehicle is a high-priority task, which is solved using the CBF techniques. The candidate CBFs can be selected as $b(D) = D-D_0$, where $D_0$ is the desired distance threshold. 
\end{itemize}
 
In this simulation, conventional CBF, robust CBF (RoCBF) \cite{jankovic2018robust}, DO-CBF \cite{Das2023} without subtracting the estimation error bound, RRCBF and DO-RRCBF are conducted for comparison. The physical parameters of ego and leading vehicles are adopted from \cite{Ames2017}. The desired cruising velocity is $v_d=20 \rm{m/s}$, the following distance threshold $D_0 = 80\rm{m}$, and the initial states are $v_l = 15{\rm{m/s}},v_e = 15 \rm{m/s}$. The acceleration of the leader vehicle is $a_l=0 {\rm m/s^2}$, and external disturbance is set as $w= \sin(t) - 0.5\sin(2t)$. The wheel force limitation is set as $u\in[-0.3mg,0.3mg]$. The control parameters are denoted as $k = 5$ for the speeding tracking controller, $\alpha_1 = \alpha_2 = 1$, $\beta = 0.01$ for CBFs and $L=10$ for the disturbance observer employed from \cite{chen2000nonlinear}.
\begin{figure}[!t]
	\centering
	\subfloat[]{
		\includegraphics[width=0.44\textwidth]{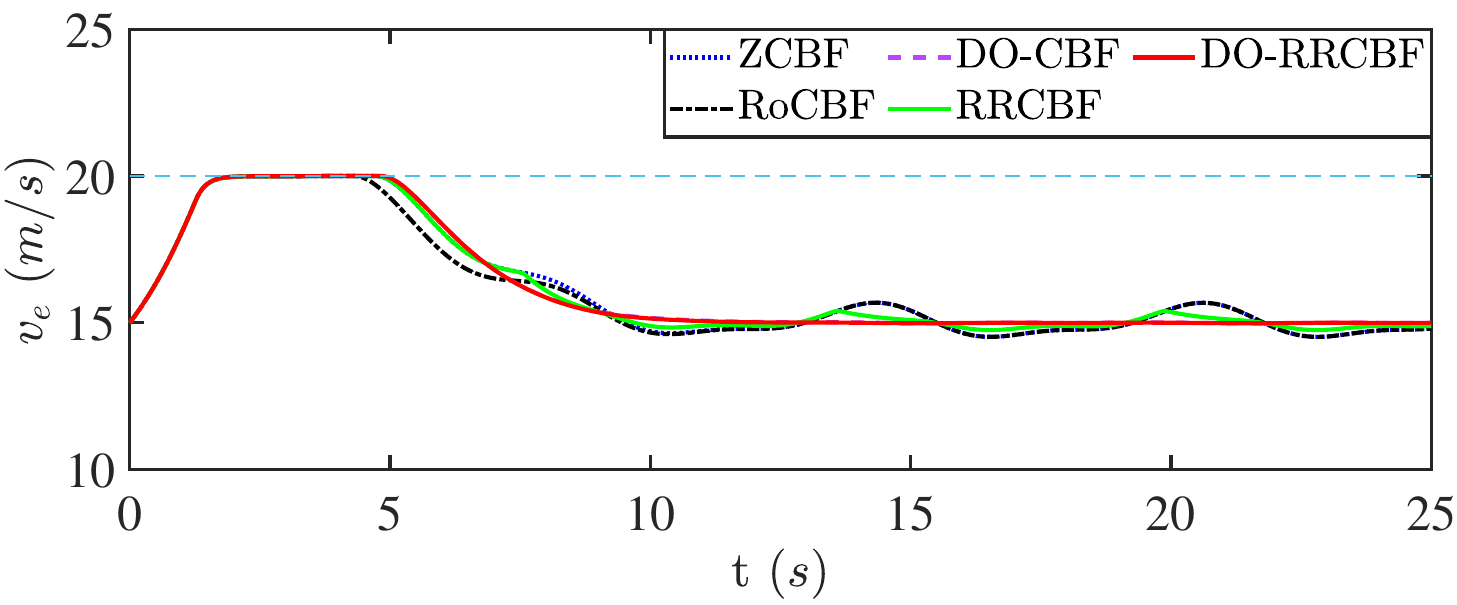}}\\
	\subfloat[]{
		\includegraphics[width=0.44\textwidth]{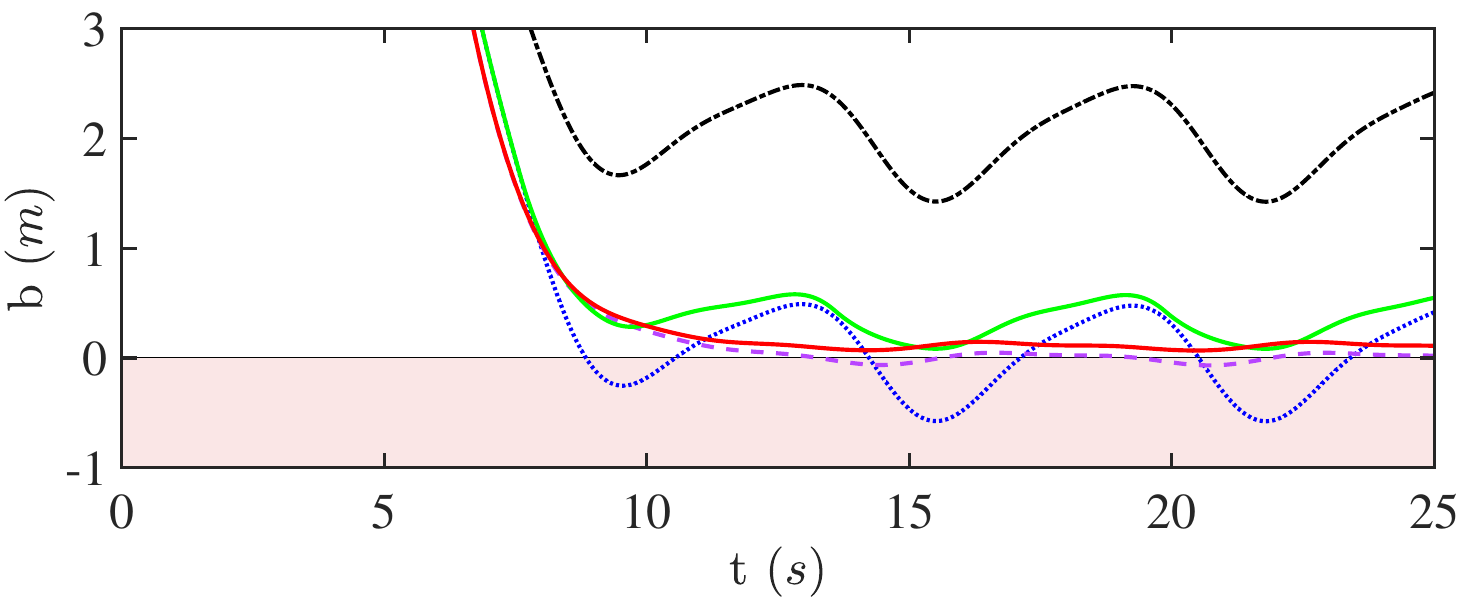}}\\
        \subfloat[]{
		\includegraphics[width=0.44\textwidth]{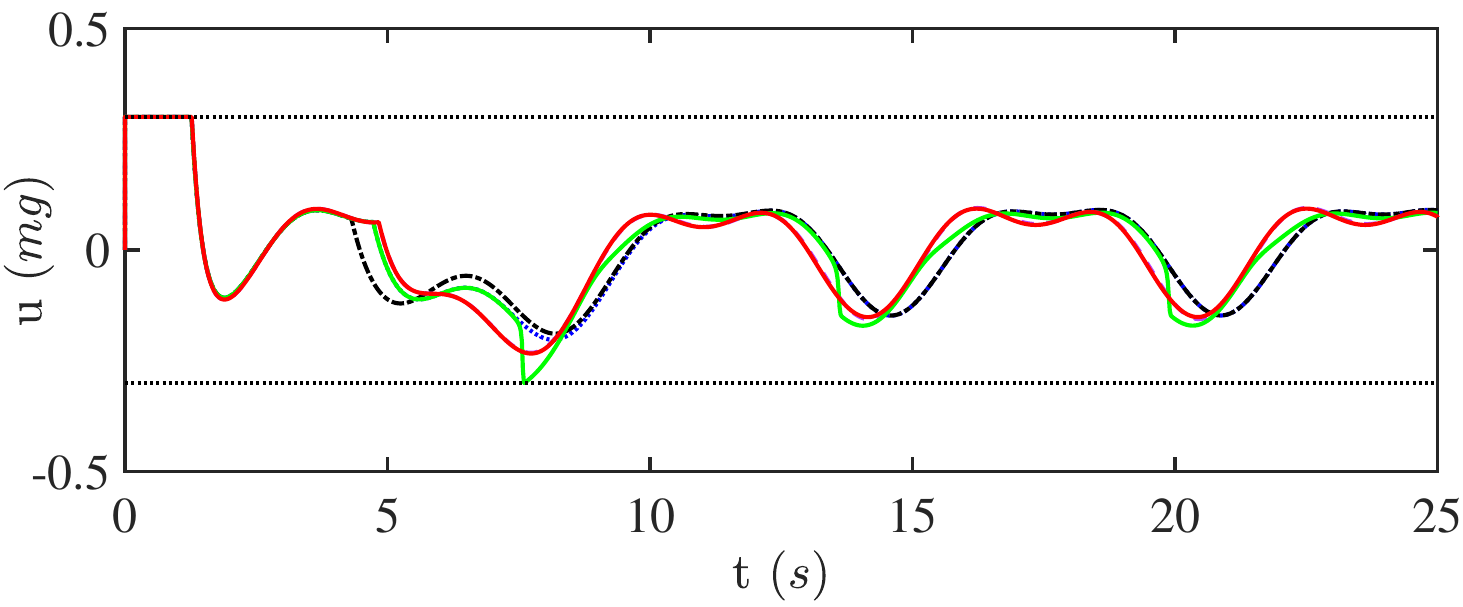}}
	\caption{Simulation results of ACC system with different CBFs: (a) Curves of ego vehicle's velocities $v_e$, (b) Curves of safety functions $b(D)$, (c) Curves of control inputs $u$.}
	\label{results2}
\end{figure} 

The ego vehicle's tracking performance under different CBF controllers is illustrated in Fig. \ref{results2}. The results show that external disturbances negatively affect safety performance; however, some tested controllers, RoCBF, RRCBF, and DO-RRCBF, successfully enforce strict safety constraints, demonstrating the strong disturbance robustness of the proposed approach. Among them, RoCBF and RRCBF maintain safety but tend to operate within slightly more conservative sets compared to DO-RRCBF. When comparing DO-CBF and DO-RRCBF, both attempt to recover nominal control performance, but DO-CBF fails to guarantee strict safety due to the absence of estimation error quantification. In contrast, DO-RRCBF not only maintains strict safety but also reduces conservatism by actively compensating the disturbance estimates, highlighting its adaptability in the presence of disturbances.
\section{Conclusions}
This note has proposed a new reciprocal resistance-based control barrier function framework to enhance safety-critical control for disturbed nonlinear systems without requiring prior knowledge of disturbance bounds or derivatives. By introducing the concept of reciprocal resistance-based barrier functions and extending it to RRCBFs, the proposed approach has ensured forward invariance and robustness against bounded disturbances. Unlike conventional robust CBF methods, the proposed RRCBF framework has eliminated the need for prior disturbance information, making it more adaptable to real-world uncertainties. Theoretical analysis and illustrative simulation tests under disturbances have validated the effectiveness of the proposed approach.

\appendix
Before introducing Lemma 3, we first define a practical version of the RRBF that incorporates a regularization term $\sigma\in\mathbb{R}_+$, replacing the potentially singular term $1/h(\boldsymbol{x})$ with $1/[h(\boldsymbol{x})+\sigma]$. It should be noted that this formulation requires prior knowledge of the disturbance bound, i.e., an upper bound on $\big \vert \frac{\partial h(\boldsymbol{x})}{\partial \boldsymbol{x}}\boldsymbol{d} \big \vert$. Note that $D\in\mathbb{R}_+$ is this upper bound. 
\begin{definition}
	Considering the system (\ref{sys_d}), a continuously differentiable function $h: \mathbb{R}^n\rightarrow\mathbb{R}$ is a practical reciprocal resistance-based barrier function if there exist extended class functions $\alpha$, $\beta$ and $\sigma\in(0,1/\beta^{-1}(D))$ subject to
	\begin{equation}
		\begin{aligned}
			\mathcal{L}_{\boldsymbol{f}}h(\boldsymbol{x}) + \alpha(h(\boldsymbol{x})) - \beta(1/[h(\boldsymbol{x})+\sigma]) \geq 0,\; \forall \boldsymbol{x}\in{\rm Int}\mathbb{C}.
		\end{aligned}
		\label{p-rrbf}
	\end{equation}
    where $\beta^{-1}(\cdot)$ is the inverse of $\beta(\cdot)$.
\end{definition}

With this definition, we state the following lemma.
\begin{lemma}
	Consider the system (\ref{sys_d}). With a valid practical RRBF $h: \mathbb{R}^n\rightarrow\mathbb{R}$, the set $\mathrm{Int}\mathbb{C}$ is forward invariant. 
\end{lemma}
\begin{proof}
    Using (\ref{p-rrbf}) and the system dynamics (\ref{sys_d}), one has 
	\begin{equation}
		\begin{aligned}
			\dot{h}(\boldsymbol{x}) &\geq -\alpha (h(\boldsymbol{x})) + \beta({1/[h(\boldsymbol{x})+\sigma]}) + \frac{\partial h(\boldsymbol{x})}{\partial \boldsymbol{x}}\boldsymbol{d}\\
            &\geq -\alpha (h(\boldsymbol{x})) + \beta({1/[h(\boldsymbol{x})+\sigma]}) - D.
		\end{aligned}
	\end{equation}
    Since $\sigma\in (0,1/\beta^{-1}(D))$ and $\big\vert \frac{\partial h(\boldsymbol{x})}{\partial \boldsymbol{x}}\boldsymbol{d} \big\vert \leq D$, it is obtained that 
    \begin{equation}
        \lim_{h(\boldsymbol{x})\rightarrow 0} \beta\bigg(\frac{1}{h(\boldsymbol{x})+\sigma}\bigg) - D = \tilde{\sigma} > 0.
    \end{equation}
     
     Define the auxiliary function $f_p(s) = \alpha(s)-\beta[1/(s+\sigma)]+D$ on $s\in\mathbb{R}_+$. We have
    \begin{equation}
        \begin{aligned}
            \lim_{s\rightarrow 0} f_p(s) = -\tilde{\sigma} < 0,\;        \lim_{s\rightarrow \infty} f_p(s) = +\infty.
        \end{aligned}
    \end{equation}
    Following a similar argument to that in Lemma 2, there exists a unique $h_p\in\mathbb{R}_+$ such that $\alpha(h_p)+D = \beta({1/[h(h_p)+\sigma]})$. Then, it is obtained that $\dot{h}(\boldsymbol{x}) \geq 0$ with $h(\boldsymbol{x})\in(0,h_p]$. Therefore, by Lemma 1, the set $\mathrm{Int}\mathbb{C}$ is forward invariant with a valid practical RRBF. This completes the proof.		  
\end{proof}

\bibliographystyle{IEEEtran}        
\bibliography{paper_bib} 

\begin{thebibliography}{10}
\providecommand{\url}[1]{#1}
\csname url@samestyle\endcsname
\providecommand{\newblock}{\relax}
\providecommand{\bibinfo}[2]{#2}
\providecommand{\BIBentrySTDinterwordspacing}{\spaceskip=0pt\relax}
\providecommand{\BIBentryALTinterwordstretchfactor}{4}
\providecommand{\BIBentryALTinterwordspacing}{\spaceskip=\fontdimen2\font plus
\BIBentryALTinterwordstretchfactor\fontdimen3\font minus
  \fontdimen4\font\relax}
\providecommand{\BIBforeignlanguage}[2]{{%
\expandafter\ifx\csname l@#1\endcsname\relax
\typeout{** WARNING: IEEEtran.bst: No hyphenation pattern has been}%
\typeout{** loaded for the language `#1'. Using the pattern for}%
\typeout{** the default language instead.}%
\else
\language=\csname l@#1\endcsname
\fi
#2}}
\providecommand{\BIBdecl}{\relax}
\BIBdecl

\bibitem{ozguner2007systems}
U.~Ozguner, C.~Stiller, and K.~Redmill, ``Systems for safety and autonomous
  behavior in cars: The {DARPA} grand challenge experience,''
  \emph{{Proceedings} of the IEEE}, vol.~95, no.~2, pp. 397--412, 2007.

\bibitem{zanchettin2015safety}
A.~M. Zanchettin, N.~M. Ceriani, P.~Rocco, H.~Ding, and B.~Matthias, ``Safety
  in human-robot collaborative manufacturing environments: {Metrics} and
  control,'' \emph{IEEE Transactions on Automation Science and Engineering},
  vol.~13, no.~2, pp. 882--893, 2015.

\bibitem{althoff2021set}
M.~Althoff, G.~Frehse, and A.~Girard, ``Set propagation techniques for
  reachability analysis,'' \emph{Annual Review of Control, Robotics, and
  Autonomous Systems}, vol.~4, no.~1, pp. 369--395, 2021.

\bibitem{rawlings2017model}
J.~B. Rawlings, D.~Q. Mayne, M.~Diehl \emph{et~al.}, \emph{Model predictive
  control: Theory, computation, and design}.\hskip 1em plus 0.5em minus
  0.4em\relax Nob Hill Publishing Madison, WI, 2017, vol.~2.

\bibitem{tee2009barrier}
K.~P. Tee, S.~S. Ge, and E.~H. Tay, ``Barrier {L}yapunov functions for the
  control of output-constrained nonlinear systems,'' \emph{Automatica},
  vol.~45, no.~4, pp. 918--927, 2009.

\bibitem{bechlioulis2010prescribed}
C.~P. Bechlioulis and G.~A. Rovithakis, ``Prescribed performance adaptive
  control for multi-input multi-output affine in the control nonlinear
  systems,'' \emph{IEEE Transactions on automatic control}, vol.~55, no.~5, pp.
  1220--1226, 2010.

\bibitem{berger2021funnel}
T.~Berger, A.~Ilchmann, and E.~P. Ryan, ``Funnel control of nonlinear
  systems,'' \emph{Mathematics of Control, Signals, and Systems}, vol.~33,
  no.~1, pp. 151--194, 2021.

\bibitem{Ames2017}
A.~D. Ames, X.~Xu, J.~W. Grizzle, and P.~Tabuada, ``Control barrier function
  based quadratic programs for safety critical systems,'' \emph{IEEE
  Transactions on Automatic Control}, vol.~62, no.~8, pp. 3861--3876, 2017.

\bibitem{guo2022control}
Z.~Guo, D.~Henry, J.~Guo, Z.~Wang, J.~Cieslak, and J.~Chang, ``Control for
  systems with prescribed performance guarantees: An alternative interval
  theory-based approach,'' \emph{Automatica}, vol. 146, p. 110642, 2022.

\bibitem{ames2014rapidly}
A.~D. Ames, K.~Galloway, K.~Sreenath, and J.~W. Grizzle, ``Rapidly
  exponentially stabilizing control {Lyapunov} functions and hybrid zero
  dynamics,'' \emph{IEEE Transactions on Automatic Control}, vol.~59, no.~4,
  pp. 876--891, 2014.

\bibitem{xu2015robustness}
X.~Xu, P.~Tabuada, J.~W. Grizzle, and A.~D. Ames, ``Robustness of control
  barrier functions for safety critical control,'' \emph{IFAC-PapersOnLine},
  vol.~48, no.~27, pp. 54--61, 2015.

\bibitem{jankovic2018robust}
M.~Jankovic, ``Robust control barrier functions for constrained stabilization
  of nonlinear systems,'' \emph{Automatica}, vol.~96, pp. 359--367, 2018.

\bibitem{alan2021safe}
A.~Alan, A.~J. Taylor, C.~R. He, G.~Orosz, and A.~D. Ames, ``Safe controller
  synthesis with tunable input-to-state safe control barrier functions,''
  \emph{IEEE Control Systems Letters}, vol.~6, pp. 908--913, 2021.

\bibitem{taylor2020adaptive}
A.~J. Taylor and A.~D. Ames, ``Adaptive safety with control barrier
  functions,'' in \emph{2020 American Control Conference (ACC)}.\hskip 1em plus
  0.5em minus 0.4em\relax IEEE, 2020, pp. 1399--1405.

\bibitem{DasarXiv}
E.~Da{\c{s}} and R.~M. Murray, ``Robust safe control synthesis with disturbance
  observer-based control barrier functions,'' in \emph{Proceedings of the
  Conference on Decision and Control (CDC)}.\hskip 1em plus 0.5em minus
  0.4em\relax IEEE, 2022, pp. 5566--5573.

\bibitem{sun2024safety}
J.~Sun, J.~Yang, and Z.~Zeng, ``Safety-critical control with control barrier
  function based on disturbance observer,'' \emph{IEEE Transactions on
  Automatic Control}, vol.~69, no.~7, pp. 4750--4756, 2024.

\bibitem{chen2023robust}
J.~Chen, Z.~Gao, and Q.~Lin, ``Robust control barrier functions for safe
  control under uncertainty using extended state observer and output
  measurement,'' in \emph{Proceedings of the IEEE Conference on Decision and
  Control (CDC)}.\hskip 1em plus 0.5em minus 0.4em\relax IEEE, 2023, pp.
  8477--8482.

\bibitem{takano2020robust}
R.~Takano and M.~Yamakita, ``Robust control barrier function for systems
  affected by a class of mismatched disturbances,'' \emph{SICE Journal of
  Control, Measurement, and System Integration}, vol.~13, no.~4, pp. 165--172,
  2020.

\bibitem{WangarXiv}
Y.~Wang and X.~Xu, ``Disturbance observer-based robust control barrier
  functions,'' in \emph{Proceedings of the American Control Conference
  (ACC)}.\hskip 1em plus 0.5em minus 0.4em\relax IEEE, 2023, pp. 3681--3687.

\bibitem{wang2024tac}
X.~Wang, J.~Yang, C.~Liu, Y.~Yan, and S.~Li, ``Safety-critical disturbance
  rejection control of nonlinear systems with unmatched disturbances,''
  \emph{IEEE Transactions on Automation Control (Early Access)}, doi:
  10.1109/TAC.2024.3496572.

\bibitem{khalil2002nonlinear}
H.~K. Khalil and J.~W. Grizzle, \emph{Nonlinear systems}.\hskip 1em plus 0.5em
  minus 0.4em\relax Prentice hall Upper Saddle River, NJ, 2002, vol.~3.

\bibitem{sontag1995characterizations}
E.~D. Sontag and Y.~Wang, ``On characterizations of the input-to-state
  stability property,'' \emph{Systems \& Control Letters}, vol.~24, no.~5, pp.
  351--359, 1995.

\bibitem{guo2017simple}
T.~Guo, Z.~Wang, X.~Wang, S.~Li, and Q.~Li, ``A simple control approach for
  buck converters with current-constrained technique,'' \emph{IEEE Transactions
  on Control Systems Technology}, vol.~27, no.~1, pp. 418--425, 2017.

\bibitem{blanchini1999set}
F.~Blanchini, ``Set invariance in control,'' \emph{Automatica}, vol.~35,
  no.~11, pp. 1747--1767, 1999.

\bibitem{brezis1970characterization}
H.~Brezis, ``On a characterization of flow-invariant sets,''
  \emph{Communications on Pure and Applied Mathematics}, vol.~23, no.~2, pp.
  261--263, 1970.

\bibitem{bartle2000introduction}
R.~G. Bartle and D.~R. Sherbert, \emph{Introduction to Real Analysis}.\hskip
  1em plus 0.5em minus 0.4em\relax John Wiley \& Sons, Inc., 2000.

\bibitem{xiao2021high}
W.~Xiao and C.~Belta, ``High-order control barrier functions,'' \emph{IEEE
  Transactions on Automatic Control}, vol.~67, no.~7, pp. 3655--3662, 2021.

\bibitem{tan2021high}
X.~Tan, W.~S. Cortez, and D.~V. Dimarogonas, ``High-order barrier functions:
  Robustness, safety, and performance-critical control,'' \emph{IEEE
  Transactions on Automatic Control}, vol.~67, no.~6, pp. 3021--3028, 2021.

\bibitem{Das2023}
E.~Da{\c{s}} and J.~W. Burdick, ``Robust control barrier functions using
  uncertainty estimation with application to mobile robots,'' \emph{IEEE
  Transactions on Automatic Control ({Early Access})}, 2025, doi:
  10.1109/TAC.2025.3538742.

\bibitem{AnilarXiv}
A.~Alan, T.~G. Molnar, E.~Das, A.~D. Ames, and G.~Orosz, ``Disturbance
  observers for robust safety-critical control with control barrier
  functions,'' \emph{IEEE Control Systems Letters}, vol.~7, pp. 1123--1128,
  2022.

\bibitem{Ames2019control}
A.~D. Ames, S.~Coogan, M.~Egerstedt, and et~al., ``Control barrier functions:
  {Theory} and applications,'' in \emph{Proceedings of the European Control
  Conference (ECC)}, 2019, pp. 3420--3431.

\bibitem{chen2000nonlinear}
W.-H. Chen, D.~J. Ballance, P.~J. Gawthrop, and J.~O'Reilly, ``A nonlinear
  disturbance observer for robotic manipulators,'' \emph{IEEE Transactions on
  Industrial Electronics}, vol.~47, no.~4, pp. 932--938, 2000.

\end{thebibliography}

\end{document}